\documentclass[USenglish, a4paper, thm-restate,numberwithinsect, cleveref]{lipics-v2021}

\usepackage{graphicx}
\usepackage{amsmath}

\usepackage{float}
\usepackage{xspace}
\usepackage{todonotes}
\usepackage{tikz}
\usepackage{algorithm}
\usepackage[noend]{algpseudocode}
\usetikzlibrary{decorations.pathmorphing}
\tikzset{snake it/.style={decorate, decoration=snake}}
\tikzset{defnode/.style={inner sep=1.6pt,circle,black,fill=darkgray}}
\tikzset{invisible/.style={draw = white, circle, fill=white, inner sep=0, minimum size=6pt}}

\newcommand{\G}{\ensuremath{\mathcal{G}}\xspace}

\xdefinecolor{green}{rgb}{0,0.7,0}

\Copyright{Daniele Carnevale, Arnaud Casteigts, Timothée Corsini}

\nolinenumbers

\title{Dismountability in Temporal Cliques Revisited}

\author{Daniele Carnevale}{Gran Sasso Science Institute, L'Aquila, Italy}{}{}{}
\author{Arnaud Casteigts}{University of Geneva, Switzerland}{}{}{}
\author{Timothée Corsini}{LaBRI, University of Bordeaux, France}{}{}{}
\authorrunning{Daniele Carnevale, Arnaud Casteigts, Timothée Corsini}



\begin{document}

\maketitle

\begin{abstract}
  A temporal graph is a graph whose edges are available only at certain points in time. It is temporally connected if the nodes can reach each other by paths that traverse the edges chronologically (temporal paths).
  Unlike static graphs, temporal graphs do not always admit small subsets of edges that preserve connectivity (temporal spanners)~--~there exist temporal graphs with $\Theta(n^2)$ edges, all of which are critical. In the case of temporal cliques (the underlying graph is complete), spanners of size $O(n\log n)$ are guaranteed. The original proof of this result by Casteigts et al. [ICALP 2019] combines a number of techniques, one of which is called \emph{dismountability}. In a recent work, Angrick \textit{et al.} [ESA 2024] simplified the proof and showed, among other things, that a one-sided version of dismountability can replace elegantly the second part of the proof (together with a recurrence argument).

  In this paper, we revisit methodically the dismountability principle. We start by characterizing the structure that a temporal clique must have if it is not $1$-hop dismountable, then neither $1$-hop nor $2$-hop (i.e. not $\{1,2\}$-hop) dismountable, and finally not $\{1,2,3\}$-hop dismountable.
  It turns out that if a clique is $k$-hop dismountable for any other $k$, then it must also be $\{1,2,3\}$-hop dismountable, thus no additional structure can be obtained beyond this point. Interestingly, excluding only $1$-hop and $2$-hop dismountability is already sufficient for reducing the spanner problem from cliques to bi-cliques. Put together with the strategy of Angrick et al. for bi-cliques, the entire $O(n \log n)$ result can now be recovered using \emph{only} dismountability. This establishes dismountability as a central technique. The existence of $O(n)$ spanners is not yet settled. However, an interesting by-product of our analysis is that any minimal counter-example to the existence of $4n$ spanners must satisfy the properties of non $\{1,2,3\}$-hop dismountable cliques, which is relevant beyond the $O(n)$ conjecture.

  In the second part, we discuss further connections between dismountability and another technique called pivotability. In particular, we show that if a temporal clique is recursively $k$-hop dismountable, then it is also pivotable (and thus admits a $2n$ spanner, whatever $k$). We also study a family of labelings (called \emph{full-range}) which force both dismountability and pivotability. The latter gives some evidence that large lifetimes could be exploited more generally for the construction of spanners.
\end{abstract}

\section{Introduction}

\noindent For most purposes, a temporal graph can be defined as a triplet $\G=(V,E,\lambda)$, where $(V,E)$ is a standard (in this work, \emph{undirected}) graph called the \emph{footprint} (or underlying graph) of $\G$, and $\lambda:E \to 2^{\mathbb{N}}$ is a labeling function that assigns a non-empty set of time labels to every edge of $E$, interpreted as availability times. A pair $(e,t)$ is a \emph{contact} of \G if $t \in \lambda(e)$. A \emph{temporal path} in $\G$ is a sequence of contacts $\langle (e_i,t_i)\rangle$ such that $\langle e_i \rangle$ is a path in the footprint and $\langle t_i \rangle$ is non-decreasing. %(a temporal path is \emph{strict} if $\langle t_i \rangle$ is increasing)
A temporal graph \G is \emph{temporally connected} ($\G$ is in class \textsf{TC}) if temporal paths exist between all ordered pairs of nodes. It is \emph{minimal} in \textsf{TC} if none of its contacts can be removed without breaking temporal connectivity.

Given a temporal graph $\G=((V,E),\lambda)$, a \emph{temporal subgraph} $\G'\subseteq \G$ is a temporal graph $\G'=((V',E'), \lambda')$ such that $V'\subseteq V$, $E' \subseteq E \cap V'\times V'$, and $\lambda'$ is a subset of $\lambda$ restricted to $E'$. If $\G \in$ \textsf{TC}, $V'=V$, and $\G' \in $ \textsf{TC}, then $\G'$ is a \emph{temporal spanner} of $\G$. Figure~\ref{fig:spanner} shows two possible temporal spanners of a same graph $\G$. One of them has the additional property that its footprint is a tree.

\begin{figure}[h]
  \centering
  \begin{tikzpicture}[xscale=1.4,yscale=.7]
    \tikzstyle{every node}=[defnode]
    \path (0,1) node (a) {};
    \path (1,2) node (b) {};
    \path (2,1) node (c) {};
    \path (1,0) node (d) {};
    \tikzstyle{every node}=[font=\footnotesize]
    \path (a) node[below=.6cm] (g){\large $\G'$};
    \path (c) node[right=.6cm] (inc){\large $\subseteq$};
    \path (a) node[left] (aa){a};
    \path (b) node[above] (bb){b};
    \path (c) node[right] (cc){c};
    \path (d) node[below] (dd){d};
    \tikzstyle{every node}=[fill=white, inner sep=1.5pt, font=\footnotesize]
    \draw (a) -- node {\color{darkgray} 2,7}(b);
    \draw (b) -- node {\color{darkgray} 3}(c);
    \draw (b) -- node {\color{darkgray} 1,4}(d);
  \end{tikzpicture}
  ~
  \begin{tikzpicture}[xscale=1.4,yscale=.7]
    \tikzstyle{every node}=[defnode]
    \path (0,1) node (a) {};
    \path (1,2) node (b) {};
    \path (2,1) node (c) {};
    \path (1,0) node (d) {};
    \tikzstyle{every node}=[font=\footnotesize]
    \path (a) node[below=.6cm] (g){\large $\G$};
    \path (a) node[left] (aa){a};
    \path (b) node[above] (bb){b};
    \path (c) node[right] (cc){c};
    \path (d) node[below] (dd){d};
    \tikzstyle{every node}=[fill=white, inner sep=1.5pt, font=\footnotesize]
    \draw (a) -- node {\color{darkgray} 2,7}(b);
    \draw (b) -- node {\color{darkgray} 3,5}(c);
    \draw (b) -- node {\color{darkgray} 1,4}(d);
    \draw (c) -- node {\color{darkgray} 1,6}(d);
    \draw (a) -- node {\color{darkgray} 4}(d);
  \end{tikzpicture}
  ~
  \begin{tikzpicture}[xscale=1.4,yscale=.7]
    \tikzstyle{every node}=[defnode]
    \path (0,1) node (a) {};
    \path (1,2) node (b) {};
    \path (2,1) node (c) {};
    \path (1,0) node (d) {};
    \tikzstyle{every node}=[font=\footnotesize]
    \path (a) node[below=.6cm] (g){\large $\G''$};
    \path (a) node[left=.6cm] (inc){\large $\supseteq$};
    \path (a) node[left] (aa){a};
    \path (b) node[above] (bb){b};
    \path (c) node[right] (cc){c};
    \path (d) node[below] (dd){d};
    \tikzstyle{every node}=[fill=white, inner sep=1.5pt, font=\footnotesize]
    \draw (a) -- node {\color{darkgray} 2}(b);
    \draw (b) -- node {\color{darkgray} 3}(c);
    \draw (c) -- node {\color{darkgray} 1}(d);
    \draw (a) -- node {\color{darkgray} 4}(d);
  \end{tikzpicture}
  \caption{\label{fig:spanner}A temporally connected graph $\G$ and two temporal spanners $\G'$ and $\G''$.}
\end{figure}

The search for small temporal spanners began in~\cite{KKK00}, where Kempe, Kleinberg, and Kumar observed that \textsf{TC} graphs do not always admit spanning trees (e.g. $\G''$ in Figure~\ref{fig:spanner}), and in fact, some of them do not even admit spanners with $O(n)$ edges (some labelings of the hypercube make all of its $\Theta(n \log n)$ edges critical for connectivity).
More than a decade later, Axiotis and Fotakis~\cite{AF16} presented an infinite family of temporal graphs having ${n \choose 2}/9+\Theta(n)=\Theta(n^2)$ edges, all of which are critical for \textsf{TC}, ruling out even the existence of $o(n^2)$-sparse spanners.

In other words, no sparse analog of the concept of a spanning tree exist in dynamic networks, a fact with deep consequences for concrete applications like routing.

\subsection{Temporal cliques}

The above negative results immediately prompt the question of whether some restricted classes of temporal graphs admit sparse spanners unconditionally. In~\cite{CPS19}, Casteigts, Peters, and Schoeters showed that if the footprint is a complete graph (the input is a \emph{temporal clique}), then $O(n \log n)$-sparse spanners are guaranteed, whatever the labeling. Quite naturally, the hardest cases are when every edge has a single label and adjacent labels are all different, i.e. the labeling is both \emph{simple} and \emph{proper}, respectively. Indeed, using more labels can only increase reachability, and so does having two adjacent edges with the same label, as the times along temporal paths must only be non-decreasing. (If one requires strict temporal paths, non-proper labelings may be worse, but this case is less interesting, as it allows arbitrary graphs to be unsparsifiable by simply having the same label on every edge~\cite{AGMS17,KKK00}.)

Thus, in the context of temporal spanners (and in the present paper), a \emph{temporal clique} refers to a temporal graph $\G=(V,E,\lambda)$ where $(V,E)$ is a complete graph and $\lambda: E \to \mathbb{N}$ is single-valued and locally injective (i.e., simple and proper). If one prefers, one may think of it as having globally unique labels, without consequences for the spanner problem. Temporal bi-cliques are defined analoguously, except that the footprint is a complete bipartite graph.

The $O(n \log n)$ result in~\cite{CPS19} was obtained by a combination of techniques. The first half of the proof involves dismountability and an algorithm called fireworks, which reduces the general problem to the particular case of finding a spanner in a temporal complete bipartite graph with significant extra structure. The second part of the proof shows that such graphs admit $O(n \log n)$ spanners using a technique called iterative delegation. In a recent article, Angrick et al.~\cite{Hasso} show that
the bipartite version of the problem not only captures the hard instances; it is actually equivalent to the original problem in the sense that the worst-case size of a spanner in cliques and such bi-cliques are related by a constant factor. They also show that the second part of the proof involving bi-cliques can be simplified by using a one-sided version of dismountability in place of iterative delegation, breaking the problem into two sub-problems recursively and recovering the $O(n \log n)$ result more simply.

As of today, the status of the problem is unsettled. In particular, it remains open whether temporal cliques always admit $O(n)$ spanners. Some sub-families of cliques do, for example cliques that are pivotable or recursively dismountable~\cite{CPS19} (more on this later), cliques resulting from encounters of agents moving along the line~\cite{HNRV23}, cliques whose labeling is random~\cite{CRRZ21}, and cliques whose labeling obeys various technical restrictions~\cite{Hasso}.

\subsection{Contributions}

In this paper, we revisit methodically the dismountability principle. In its default version, $1$-hop dismountability, this principle says that if a node $v$ is both the latest neighbor of some node and the earliest neighbor of another node, then $v$ can be removed from the instance at the cost of including two edges in the spanner, and one can subsequently recurse on $\G \setminus \{v\}$. This technique can be generalized to $k$-hop dismountability, when a similar relation applies between $v$ and other nodes through temporal paths of length at most $k$. In this case, the number of edges added to the spanner before recursion is at most $2k$, which implies that minimal counter-examples to the existence of $2kn$ spanners cannot be $k$-hop dismountable.

We start by characterizing the structure of temporal cliques which are not $1$-hop dismountable, then neither $1$-hop nor $2$-hop dismountable (i.e., not $\{1,2\}$-hop dismountable), and finally not $\{1,2,3\}$-hop dismountable. This structure is characterized in a necessary and sufficient way. Quite surprisingly, we show that if a clique is $k$-hop dismountable for $k>3$, then it must also be $k$-hop dismountable for some $k\le 3$, which implies that one does not need to analyse the structure of non $k$-hop dismountability beyond this point. For technical reasons, our analysis actually implies that any counter example to the existence of $4n$ (not just $6n$) spanners must possess all the characterized properties.

If one does not care about the multiplicative constant within $O(n)$ spanners, then it turns out that excluding only $1$-hop and $2$-hop dismountability is already sufficient for reducing the spanner problem from cliques to bi-cliques. This reduction replaces the fireworks algorithm from~\cite{CPS19} with much simpler arguments, and unlike~\cite{Hasso}, it does so while preserving the number of vertices, which may be useful in future studies about the precise size of a spanner.

More generally, the main outcome of our study is that the entire $O(n \log n)$ result can now be recovered using dismountability only, with a simple and unified constructive algorithm, establishing dismountability as a central technique for the problem.

In the second part of the paper, we further explore the properties of (non-)dismountability, some of its possible relaxations and its connections to other techniques, including pivotability. In particular, we show that if a temporal clique is recursively $k$-hop dismountable, then it is also pivotable.
% We also strengthen a conjecture from~\cite{CPS19}, suggesting that temporal cliques may admit $2n-3$ spanners whose footprint is a $2$-arch graph.
We also define the concept of a \emph{(time-)compressed} temporal graph and identify a new family of labelings called \emph{full-range}, which corresponds to compressed temporal graphs whose lifetime (interval of existence) is equal to the number of edges. We prove that full-range cliques are $\{1,2,3\}$-hop dismountable, and full-range temporal graphs in \textsf{TC} (not necessarily cliques, but including all temporal cliques) are pivotable. These results give some evidence that having a large lifetime could be helpful for the construction of spanners.
% MORE Finally, we suggest a general strategy for targetting $O(n)$-size spanners in the general case, based on a recurrence relation inspired from~\cite{Hasso}. (TO KEEP OR REMOVE, WE'LL SEE.)

\subsection{Further related works}
Beyond structural results, temporal spanners have attracted significant attention over the last decade. On the algorithmic side, the problem of minimizing the number of labels of a temporal spanner was shown \textsf{APX}-hard in the ``non-simple, non-proper, strict'' setting~\cite{AGMS17} and in the ``simple, non-proper, non-strict'' setting~\cite{AF16}, both settings being incomparable in terms of expressivity~\cite{CCS24}. Observe that in the second setting, the problem coincides with that of minimizing the number of edges (as the graph is simple), so this version of the problem is~\textsf{APX}-hard as well. The status of the problem is open for general temporal graphs which are both simple and proper. Finally, deciding if a given temporal graph admits a temporal spanning tree is \textsf{NP}-complete in the ``non-simple, proper'' setting (regardless of whether strict or non-strict paths are considered)~\cite{CC24}.

Observe that most of these works are only concerned with the \emph{size} of the spanners, which contrasts with typical spanner problems in static graphs, where solutions of size $n-1$ (spanning trees) are \emph{de facto} guaranteed and further criteria are thus considered (typically, weights and tradeoffs between density and stretch factors, see e.g.~\cite{static1,static2,static3,static4}). The reader interested additional criteria in a temporal setting is referred to~\cite{BDG+22a} (stretch factors), \cite{BDG+22b} (fault-tolerance), and~\cite{branchings} (timing criteria).

  On the structural side, a few other related results were obtained. In~\cite{CRRZ21}, Casteigts, Raskin, Renken, and Zamaraev show that any counter-example family to linear-size spanners (such as the graphs from~\cite{KKK00} and in~\cite{AF16}) are statistically insignificant. Namely, temporal Erdös-Rényi graphs almost surely admit spanners of size $2n+o(n)$ as soon as they become temporally connected (at $p=(3+\epsilon)\log n/n$).
  Interestingly, results from the 80-90s in gossip theory can be re-interpreted in the framework of temporal graphs, shedding light on the structure of minimal temporal graphs in \textsf{TC}. For example, the characterization of minimal information flows in~\cite{labahn93} implies a similar quadratic bound as the infinite family from~\cite{AF16}, once interpreted in the setting of temporal graphs.

  Finally, some recent work in network creation games directly exploit the existence of $O(n\log n)$-sparse spanners in temporal cliques~\cite{game-theory}, in a game-theoretic context.

  \subsection{Organization of the document}

  The document is organized as follows. In Section~\ref{sec:dismountability}, we recall the basic definitions of $1$-hop and $k$-hop dismountability. In Section~\ref{sec:non-dismountability}, we characterize the structure of non $k$-hop dismountable graphs for $k\le 3$ and we show that considering larger values of $k$ is unnecessary. In Section~\ref{sec:spanners}, we summarize the overall proof of existence of $O(n \log n)$ spanner in temporal cliques, using only dismountability arguments. In Section~\ref{sec:misc}, we provide further results about dismountability, pivotability, and full-range temporal graphs. We conclude in Section~\ref{sec:conclusion} with some remarks and open questions.

  \section{Dismountability}
  \label{sec:dismountability}

To start, let us recall the principle of dismountability defined in~\cite{CPS19}.
Let $\G=(V,E,\lambda)$ be a (simple and proper) temporal clique. For any node $v\in V$, let us denote by $e^-(v)$ the earliest edge incident to $v$ and by $n^-(v)$ the corresponding neighbor (its earliest neighbor).
We define similarly the latest edge $e^+(v)$ and latest neighbor $n^+(v)$.

Because $\G$ is a clique, the fact that $n^-(v)=u$ for some $u$ implies that $u$ can reach every other node through $v$ (indeed, $v$ still has a later edge with every other node after time $\lambda(uv)$). Likewise, if $n^+(w)=u$, then every node can reach $u$ through $w$.
Let $u,v,w$ be three nodes that satisfy $n^-(v)=u$ and $n^+(w)=u$. Then $u$ can {\em delegate} its emissions to $v$, and delegate its receptions to $w$.
These observations suggest a possible approach for self-reducing the problem of finding a temporal spanner as follows:

\begin{theorem}[$1$-hop dismountability~\cite{CPS19}]
  \label{thm:dismountability}
  Let $\G$ be a temporal clique, and let $u, v, w$ be three nodes such that $n^-(v)=u$ and $n^+(w)=u$. Let $S'$ be a temporal spanner of $\G\setminus \{u\}$ %\G[V\setminus u]$.
  Then $S:=S'\cup \{uv, uw\}$ is a temporal spanner of $\G$.
\end{theorem}

In this case, we say that node $u$ is $1$-hop dismountable, and by extension, a graph is $1$-hop dismountable if at least one node is $1$-hop dismountable. Such a node may or may not exist, and the recursion could thus be stuck at some point. A graph is {\em recursively $1$-hop dismountable} if one can find an ordering of $V$ that allows for a complete recursive $1$-hop dismounting of the graph (down to $n=2$).
An example of recursively $1$-hop dismountable graph, together with the resulting spanner, is shown in Figure~\ref{fig:full_dismountability}.

\begin{figure}[ht]
\centering
     \begin{subfigure}[b]{0.22\textwidth}
          \centering
          \resizebox{\linewidth}{!}{
          \centering
          \begin{tikzpicture}[scale=2.5]
    \path (0,0) coordinate (m);
    \tikzstyle{every node}=[circle,fill=white,inner sep=1.3pt]
    \path (m)+(90:1) node[defnode, fill=gray] (a){};
    \path (m)+(90-72:1) node[defnode] (b){};
    \path (m)+(90-144:1) node[defnode] (c){};
    \path (m)+(90+144:1) node[defnode] (d){};
    \path (m)+(90+72:1) node[defnode] (e){};

    \tikzstyle{every node}=[circle,fill=white,inner sep=1.3pt,font=\LARGE]
    \tikzstyle{every path}=[thin,gray,shorten >= 5pt, shorten <= 5pt]
    \tikzstyle{chain}=[ultra thick, black,shorten <= 5pt, shorten >= 5pt]
    \draw (a) edge[chain] node[pos=.5]{9}(e);
    \draw (a)-- node[pos=.5]{0}(b);
    \draw (a) edge[chain] node[pos=.5]{3}(c);
    \draw (a)-- node[pos=.5]{2}(d);
    \draw (b) --node[pos=.5]{4}(c);
    \draw (b) -- node[pos=.5]{6}(d);
    \draw (b) -- node[pos=.5]{1}(e);
    \draw (c) -- node[pos=.5]{5}(d);
    \draw (c)-- node[pos=.5]{7}(e);
    \draw (d) --node[pos=.5]{8}(e);

  \end{tikzpicture}
          }
          %\caption{Dismountable node in $\mathcal{G} = \mathcal{K}_5$.}
          \label{fig:first_dism}
     \end{subfigure}
     \begin{subfigure}[b]{0.22\textwidth}
          \centering
          \resizebox{\linewidth}{!}{
          \centering
          \begin{tikzpicture}[scale=2.5]
    \path (0,0) coordinate (m);
    \tikzstyle{every node}=[circle,fill=white,inner sep=1.3pt]
    \path (m)+(90-72:1) node[defnode] (b){};
    \path (m)+(90-144:1) node[defnode] (c){};
    \path (m)+(90+144:1) node[defnode] (d){};
    \path (m)+(90+72:1) node[defnode, fill=gray] (e){};

    \tikzstyle{every node}=[circle,fill=white,inner sep=1.3pt,font=\LARGE]
    \tikzstyle{every path}=[thin,gray,shorten >= 5pt, shorten <= 5pt]
    \tikzstyle{chain}=[ultra thick, black,shorten <= 5pt, shorten >= 5pt]
    \draw (b) --node[pos=.5]{4}(c);
    \draw (b) -- node[pos=.5]{6}(d);
    \draw (b) edge[chain] node[pos=.5]{1}(e);
    \draw (c) -- node[pos=.5]{5}(d);
    \draw (c) edge[chain] node[pos=.5]{7}(e);
    \draw (d) --node[pos=.5]{8}(e);

  \end{tikzpicture}
          }
          %\caption{Dismountable node in $\mathcal{G}'$.}
          \label{fig:second_dism}
     \end{subfigure}
     \begin{subfigure}[b]{0.2\textwidth}
          \resizebox{\linewidth}{!}{
          \begin{tikzpicture}[scale=3]
    \path (0,0) coordinate (m);
    \tikzstyle{every node}=[circle,fill=white,inner sep=1.3pt]
    \path (m)+(90-72:1) node[defnode, fill=gray] (b){};
    \path (m)+(90-144:1) node[defnode] (c){};
    \path (m)+(90+144:1) node[defnode] (d){};

    \tikzstyle{every node}=[circle,fill=white,inner sep=1.3pt,font=\LARGE]
    \tikzstyle{every path}=[thin,gray,shorten >= 5pt, shorten <= 5pt]
    \tikzstyle{chain}=[ultra thick, black,shorten <= 5pt, shorten >= 5pt]
    \draw (b) edge[chain] node[pos=.5]{4}(c);
    \draw (b) edge[chain] node[pos=.5]{6}(d);
    \draw (c) edge node[pos=.5]{5}(d);

  \end{tikzpicture}
          }
          %\caption{Select all edges in $\mathcal{G}'' = \mathcal{K}_3$.}
          \label{fig:K3}
     \end{subfigure}
     \hspace{20pt}
     \begin{subfigure}[b]{0.22\textwidth}
          \resizebox{\linewidth}{!}{
          \begin{tikzpicture}[scale=2.5]
    \path (0,0) coordinate (m);
    \tikzstyle{every node}=[circle,fill=white,inner sep=1.3pt]
    \path (m)+(90:1) node[defnode] (a){};
    \path (m)+(90-72:1) node[defnode] (b){};
    \path (m)+(90-144:1) node[defnode] (c){};
    \path (m)+(90+144:1) node[defnode] (d){};
    \path (m)+(90+72:1) node[defnode] (e){};

    \tikzstyle{every node}=[circle,fill=white,inner sep=1.3pt,font=\LARGE]
    \tikzstyle{every path}=[shorten >= 5pt, shorten <= 5pt]
    \tikzstyle{chain}=[ultra thick, shorten <= 5pt, shorten >= 5pt]
    \draw [dashed, gray](a) -- node[pos=.5]{0}(b);
    \draw (a)edge[chain] node[pos=.5]{3}(c);
    \draw [dashed, gray](a)-- node[pos=.5]{2}(d);
    \draw (a) edge[chain] node[pos=.5]{9}(e);
    \draw (b) edge[chain] node[pos=.5]{4}(c);
    \draw (b) edge[chain] node[pos=.5]{6}(d);
    \draw (b) edge[chain] node[pos=.5]{1}(e);
    \draw (c) edge[chain] node[pos=.5]{5}(d);
    \draw (c) edge[chain] node[pos=.5]{7}(e);
    \draw [dashed, gray](d) -- node[pos=.5]{8}(e);

  \end{tikzpicture}
          }
          %\caption{Obtained temporal spanner $\mathcal{S}(\mathcal{G})$ .}
          \label{fig:full_dism_spanner}
     \end{subfigure}
     \caption{\label{fig:full_dismountability} A recursively $1$-hop dismountable graph and the resulting spanner (Fig. from~\cite{CPS19}).}
 \end{figure}

 Observe that exactly $2$ edges are included in the spanner in every dismounting step, which results in a spanner of size $2n-3$ if the graph is recursively $1$-hop dismountable.

\subsection{$k$-hop dismountability}
The dismountability technique can be generalized to multi-hop temporal paths. Here, we recall the ideas, referring again the reader to~\cite{CPS19} for more details. Let $P^-$ be a temporal path from some node $u$ to some node $v$, such that this path arrives at $v$ through $e^-(v)$, then $u$ can delegate its emissions to $v$ in the same way as for $1$-hop dismountability. The edges along the path may or may not be locally earliest, what matters is only that the edge incident to $v$ is.
Similarly, let $P^+$ be a temporal path from some node $w$ to $u$ that departs from $w$ through $e^+(w)$, then $w$ can reach $u$ after having been reached by all the other nodes. As before, we can now select the edges of both paths in the spanner, remove $u$ (and only $u$), and recurse in $\G':=\G \setminus \{u\}$. Whatever spanner exist in $\G'$ plus the edges of the two paths will be a valid spanner of $\G$. The fact that one can recurse ``obliviously'' here (i.e. without taking into account which specific paths were used) is less immediate than before; in particular, the edge $e^-(v)$ may still exist in $\G \setminus \{u\}$ and $v$ could use it for reaching other nodes instead of departing from a local edge after $\lambda(e^-(v))$. But in this case, a path from $u$ to that node would still be preserved through $n^-(v)$.

If both paths have length $\le k$ for some $k$, then $u$ is $k$-hop dismountable (and by extension, the graph is $k$-hop dismountable).
See Figure~\ref{fig:dismountability} for an illustration.

\begin{figure}[h]
  \centering
\begin{subfigure}[b]{0.30\textwidth}
  \centering
\begin{tikzpicture}[scale=2]
    \path (0,0) coordinate (m);
    \draw (0,0) ellipse (1cm and .5cm) node[invisible,anchor=south] {};
    \tikzstyle{every node}=[circle,fill=white,inner sep=1.3pt]
    \path (m)+(-.5,0) node[defnode] (v) {};
    \path (m)+(.5,0) node[defnode] (w) {};
    \path (m)+(0,-.7) node[defnode] (u){};
    \draw (v)+(30:.15) node[invisible,anchor=north] {$v$};
    \draw (u)+(30:.15) node[invisible,anchor=north] {$u$};
    \draw (w)+(30:.15) node[invisible,anchor=north] {$w$};
    \tikzstyle{every node}=[fill=white,inner sep=0pt]
    \tikzstyle{every path}=[shorten >= 1pt, shorten <= 1pt]

    \draw (u)-- node[pos=.8,left=2.2pt]{\small $-$}(v);
    \draw (u)-- node[pos=.8,right=2.2pt]{\small $+$}(w);
  \end{tikzpicture}
\caption{$1$-hop dismountability.}
\end{subfigure}
\begin{subfigure}[b]{0.38\textwidth}
  \centering
\begin{tikzpicture}[scale=1.5]
    \draw (0,0) ellipse (1.25cm and .75cm) node[invisible,anchor=south] {};
    \tikzstyle{every node}=[circle,fill=white,inner sep=1.3pt]
    \path (0,-1) node[defnode] (u) {};
    \path (-.6,-.3) node[defnode] (passby2) {};
    \path (.6,-.3) node[defnode] (passby) {};
    \path (-.6,.4) node[defnode] (v) {};
    \path (.6,.4) node[defnode] (w) {};
    \draw (v) node[right=3pt,invisible] {$v$};
    \draw (u) node[right=3pt,invisible] {$u$};
    \draw (w) node[right=3pt,invisible] {$w$};
    \tikzstyle{every node}=[circle,fill=white,inner sep=0pt]
    \tikzstyle{every path}=[shorten >= 1pt, shorten <= 1pt]

    \draw (w)-- node[rectangle,inner sep=.5pt,pos=.25,right=2pt]{\footnotesize $+$}(passby);
    \draw (passby)-- (u);
    \draw (passby2)-- node[rectangle,inner sep=0.5pt,pos=.75,left]{\footnotesize $-$}(v);
    \draw (u)-- (passby2);
    \draw[gray!80,semithick] (u)+(-.2,.4) edge[->,shorten <= 0pt, shorten >= 10pt, bend left=20] ([xshift=5pt,yshift=-2pt]v);
    \draw[gray!80,semithick] (u)+(.2,.4) edge[<-,shorten <= 0pt, shorten >= 10pt, bend right=20] ([xshift=-5pt,yshift=-2pt]w);
  \end{tikzpicture}
\caption{Example of $2$-hop dismountability.}
\end{subfigure}
\caption{\label{fig:dismountability}Illustration of the principle of dismountability and $k$-hop dismountability, where $-$ and $+$ denote the local earliest and latest edges, respectively.}
\end{figure}

When $k$-hop dismounting a node, the number of edges to be included in the spanner is at most $2k$. Thus, if a graph is recursively $k$-hop dismountable for some \emph{constant} $k$, then this graph admits a spanner of size $O(n)$. Otherwise, the recursion stops at a point where the clique must have additional structure, which we characterize completely in this paper.

\section{Structure of non-dismountable graphs}
\label{sec:non-dismountability}

In this section, we characterize the properties that a temporal clique must have if it is not $1$-hop dismountable. We then proceed similarly for cliques that are neither $1$-hop nor $2$-hop dismountable (i.e., not $\{1,2\}$-hop dismountable), and finally not $\{1,2,3\}$-hop dismountable. Then, we prove that if a clique is $k$-hop dismountable for $k>3$, then it must also be $\{1,2,3\}$-hop dismountable, which implies that no further structure is to be expected for larger values of $k$. Our characterization is necessary and sufficient.

% As a result, cliques that are recursively $k$-hop dismountable \emph{for any} $k$ admit a spanner of size $O(n)$.

\subsection{Non 1-hop dismountable cliques}

Let $V^-=\{n^-(v) : v\in V\}$ be the nodes that are the earliest neighbor of at least one node, and $V^+=\{n^+(v) : v\in V\}$ be the nodes that are the latest neighbor of at least one node. Finally, let $V^0=V\setminus (V^- \cup V^+)$ be the nodes that are neither in $V^-$ nor in $V^+$.

\begin{lemma}
  If $\G$ is not $1$-hop dismountable, then $\{V^-, V^+, V^0\}$ form a partition of $V$.
\end{lemma}
\begin{proof}
  By definition, $V^0 = V \setminus (V^- \cup V^+)$. Moreover, if $V^- \cap V^+ \ne \emptyset$, then any node in the intersection is 1-hop dismountable.
\end{proof}

\subsection{Non \{1,2\}-hop dismountable cliques}

% From \cref{lemma2}, we have $V^-\cap V^+ = \emptyset$.

% \begin{corollary} \label{cor1}
%   $V^-\subseteq L^+$ and $V^+\subseteq L^-$.
% \end{corollary}

% \begin{tikzpicture}[scale=.5]
%   % \draw[fill opacity=0.25, fill=red] (0,0) circle (2cm);
%   % \draw[fill opacity=0.25, fill=green] (3,0) circle (2cm);
%   % \draw[fill opacity=0.25, fill=green] (-.5,0) circle (1cm);
%   % \draw[fill opacity=0.25, fill=green] (3.5,0) circle (1cm);
%   \draw[fill opacity=0.25, fill=red] (-2,-2) rectangle (5,2);
%   \draw[fill opacity=0.25, fill=blue] (-2,-2) rectangle (5,2);
%   \draw[fill=white] (0,0) circle (1cm);
%   \draw[fill opacity=0.25, fill=blue] (0,0) circle (1cm);
% \end{tikzpicture}
% \begin{tikzpicture}[scale=.5]
%   % \draw[fill opacity=0.25, fill=red] (0,0) circle (2cm);
%   % \draw[fill opacity=0.25, fill=green] (3,0) circle (2cm);
%   % \draw[fill opacity=0.25, fill=green] (-.5,0) circle (1cm);
%   % \draw[fill opacity=0.25, fill=green] (3.5,0) circle (1cm);
%   \draw[fill opacity=0.25, fill=red] (-2,-2) rectangle (5,2);
%   \draw[fill opacity=0.25, fill=blue] (-2,-2) rectangle (5,2);
%   \draw[fill=white] (3,0) circle (1cm);
%   \draw[fill opacity=0.25, fill=red] (3,0) circle (1cm);
% \end{tikzpicture}

\begin{lemma}\label{lem:simple-2-hop}
  Let $u$ and $v$ be two nodes of $V^+$. If $\G$ is not $\{1,2\}$-hop dismountable, then $n^-(u)$ and $n^-(v)$ must be different. Similarly, if $u$ and $v$ are two nodes of $V^-$, then $n^+(u)$ and $n^+(v)$ must be different.
\end{lemma}
\begin{proof}
  We prove the first statement, the other being symmetric.
  Let $u,v \in V^+$ and $n^-(u)=n^-(v)=w$. Without loss of generality, assume that $\lambda(v,w) < \lambda(u,w)$ (see~\cref{fig:simple-2-hop}). Node $v$ can reach $u$ in two hops, arriving through $e^-(u)$, and because it is in $V^+$, it can also be reached by some node $z$ through $e^+(z)$. Node $v$ is thus 2-hop dismountable. Note that $z$ could be in $V^0$, $V^-$, or $V^+$ (it could even be $u$), without affecting the lemma (it cannot be $w$, though, since $\lambda(v,w) < \lambda(u,w)$).
\end{proof}

\begin{figure}[h]
    \begin{center}
\begin{tikzpicture}[xscale=.8]
  \tikzstyle{every node}=[inner sep=2pt,circle,black,fill=darkgray]
  \path (1,1) node (w){};
  \path (0,0) node (u){};
  \path (2,0) node (v){};
  \path (3,1) node (z){};
  % \path (-1,1) node (x){};

  \tikzstyle{every node}=[]
  \draw (v) node[draw,circle,inner sep=2.7,thick,brown](vv){};
  \path (w) node[above=2pt]{$w\in V^-$};
  \path (u) node[below=2pt]{$u\in V^+$};
  \path (v) node[below=2pt]{$v\in V^+$};
  % \path (x) node[above=2pt]{$x$};
  \path (z) node[above=2pt]{$z$};
  \tikzstyle{every node}=[font=\scriptsize]
  \draw (u)-- node[left=-1.5pt,pos=.25]{$-$}(w);
  \draw (v)-- node[right=-2pt,pos=.25]{$-$}(w);
  % \draw (x)-- node[right=-1pt,pos=.1]{$+$}(u);
  \draw (z)-- node[left=-1pt,pos=.1]{$+$}(v);
  \draw[rounded corners,->,gray,thick] ([xshift=-5pt]v.west) -- ([yshift=-5pt]w.south) -- ([xshift=5pt]u.east);
\end{tikzpicture}
\caption{\label{fig:simple-2-hop}Node $v$ is 2-hop dismountable.}
\end{center}
\end{figure}

\begin{corollary}\label{cor:matching}~
  %\cref{lem:simple-2-hop} has the following consequences:
  \begin{enumerate}
    %\{v,n^-(v)\}
  \item The edges $\{e^-(v): v\in V^+\}$ form a perfect matching between $V^-$ and $V^+$.
  \item The edges $\{e^+(v): v\in V^-\}$ form a perfect matching between $V^-$ and $V^+$.
  \item $V^-$ and $V^+$ have the same size.
  \end{enumerate}
\end{corollary}

\noindent We now establish that $V^0$ must be empty.

\begin{lemma}\label{lemma5}
  If $\G$ is not $\{1,2\}$-hop dismountable, then $V^0 = \emptyset$.
\end{lemma}
\begin{proof} We will prove the contrapositive.
Let $w\in V^0$, let $v_1=n^-(w) \in V^-$, and let $v_2=n^+(w) \in V^+$.
By~\cref{cor:matching}, there is another node $u_1\in V^+$ such that $n^-(u_1)=v_1$, and another node $u_2 \in V^-$ such that $n^+(u_2)=v_2$. Since $u_1\in V^+$, there is also a node $x$ such that $n^+(x)=u_1$ and a node $y$ such that $n^-(y)=u_2$. Let $a=\lambda(u_1,v_1)$, $b=\lambda(v_1,w)$, $c=\lambda(w,v_2)$, and $d=\lambda(v_2,u_2)$ (see~\cref{fig:complex-2-hop}).
\begin{figure}[h]
    \begin{center}
\begin{tikzpicture}[xscale=.8]
  \tikzstyle{every node}=[inner sep=2pt,circle,black,fill=darkgray]
  \path (1,1) node (v1){};
  \path (0,0) node (u1){};
  \path (4,0) node (u2){};
  \path (2,0) node (w){};
  \path (3,1) node (v2){};
  \path (-1,1) node (x){};
  \path (5,1) node (y){};

  \tikzstyle{every node}=[]
  \path (v1) node[above=2pt]{$v_1 \in V^-$};
  \path (u1) node[below=2pt]{$u_1\in V^+$};
  \path (u2) node[below=2pt]{$u_2\in V^-$};
  \path (w) node[below=2pt]{$w\in V^0$};
  \path (x) node[above=2pt]{$x$};
  \path (y) node[above=2pt]{$y$};
  \path (v2) node[above=2pt]{$v_2\in V^+$};
  \tikzstyle{every node}=[font=\scriptsize]
  \draw (u1)-- node[right=-1.5pt,pos=.1]{$-$} node[above left=-2pt]{$a$}(v1);
  \draw (w)-- node[left=-2pt,pos=.1]{$-$} node[above right=-2pt]{$b$}(v1);
  \draw (w)-- node[right=-1pt,pos=.1]{$+$} node[above left=-2pt]{$c$}(v2);
  \draw (u2)-- node[left=-2pt,pos=.1]{$+$} node[above right=-2pt]{$d$}(v2);
  \draw (x)-- node[below left=-3pt,pos=.1]{$+$}(u1);
  \draw (y)-- node[below right=-3pt,pos=.1]{$-$}(u2);
\end{tikzpicture}
\caption{\label{fig:complex-2-hop}Illustration of the proof of Lemma~\ref{lemma5}.}
\end{center}
\end{figure}
Let us examine the relation between $a,b,c,$ and $d$.
    If $a<b$ then $u_1$ is 2-hop dismountable with respect to $x$ and $w$. If $c<d$ then $u_2$ is 2-hop dismountable with respect to $y$ and $w$. We thus have $a > b$ and $c>d$. But in this case $w$ is 2-hop dismountable with respect to $u_1$ and $u_2$. (As before, some of these nodes may coincide, e.g. $u_1$ and $v_2$, without affecting the lemma.)
\end{proof}

% \begin{corollary}
%   Temporal cliques having an odd number of nodes are $\{1,2\}$-hop dismountable.
% \end{corollary}
% \begin{proof}
%   Follows from Corollary~\ref{cor:matching} (3) and Lemma~\ref{lemma5}.
% \end{proof}

Thus, $V^-$ and $V^+$ form a partition of $V$.

\begin{lemma}[Reciprocity]\label{cor:reciprocity}
    Let $v\in V^+$ and $u\in V^-$ such that $u=n^-(v)$. If $\G$ is not $\{1,2\}$-hop dismountable, then $v$ is also the earliest neighbor of $u$ among the nodes of $V^+$.
    Similarly, if $v=n^+(u)$, then $u$ is the latest neighbor of $v$ among the nodes of $V^-$.
\end{lemma}

\begin{proof}
  We will prove the contrapositive. Let $w \in V^+$ such that $\lambda(uw)<\lambda(uv)$. Then, $w$ can reach $v$ in two hops, arriving at $v$ through $e^-(v)$. Being in $V^+$, it can also be reached by a node $x\in V^-$ through $e^+(x)$. Thus, $w$ is 2-hop dismountable with respect to $v$ and $x$. The second part of the statement is symmetrical.
\end{proof}

We insist that reciprocity does not imply that the ``earliest neighbor'' relation (or the latest neighbor relation) is symmetric, it only implies that it is the case if we restrict the graph to the edges between $V^-$ and $V^+$. In fact, node $n^-(v)$ must have earlier neighbors than $v$ within its own part $V^-$ (and likewise for later neighbors of $n^+(v)$ within $V^+$), as otherwise the graph would be $1$-hop dismountable. To avoid ambiguities, we depict the $+$ or $-$ signs with different sizes (large size for the actual earliest/latest neighbor, and small size when restricted to the other part), as shown in Figure~\ref{fig:reciprocity}.

\begin{figure}[h]
    \centering
    \begin{tikzpicture}[scale=2]
      \tikzstyle{every node}=[inner sep=2pt,circle,black,fill=darkgray]
      \path (0,0) node (u){};
      \path (1,-.6) node (v){};
      \path (0,-.6) node (y){};
      \path (1,0) node (x){};

      \tikzstyle{every node}=[]
      \path (u) node[left=2pt]{$u$};
      \path (v) node[right=2pt]{$v$};
      \path (y) node[left=2pt]{$y$};
      \path (x) node[right=2pt]{$x$};
      \path (-.6,-.3) node{$V^-$};
      \path (1.6,-.3) node{$V^+$};
      \draw (-.1,-.3) ellipse (.25cm and .6cm);
      \draw (1.1,-.3) ellipse (.25cm and .6cm);

      \draw (x)-- (y);
      \draw (u)-- node[above=-2pt,pos=.16]{\Large $+$} node[above=-2pt,pos=.84]{\scriptsize $+$} (x);
      \draw (y)-- node[below=-2pt,pos=.18]{\scriptsize $-$} node[below=-2pt,pos=.84]{\Large $-$} (v);
    \end{tikzpicture}

    \caption{Reciprocity in a non $\{1,2\}$-hop dismountable graph.}
    \label{fig:reciprocity}
\end{figure}

At this point, the reader interested only in the quest for finding linear spanners (whatever the constant) and/or in a proof of existence of $O(n \log n)$ spanners can go directly to Section~\ref{sec:spanners}.

\subsection{Non \{1,2,3\}-hop dismountable cliques}

\begin{lemma}\label{lemma10}
    Let $\G=(V,E,\lambda)$ be a non $\{1,2\}$-hop dismountable clique. Let $x,y\in V^+$ (respectively, $x,y\in V^-$), let $e=\{x,u\}=e^-(x)$ and $f=\{y,v\}=e^-(y)$ (respectively, $e^+(x)$ and $e^+(y)$). If $\G$ is not 3-hop dismountable then $\lambda(uv)$ cannot be between $\lambda(ux)$ and $\lambda(vy)$.
\end{lemma}

\begin{figure}[h]
    \centering
    \begin{tikzpicture}[scale=2]
      \tikzstyle{every node}=[inner sep=2pt,circle,black,fill=darkgray]
      \path (0,0) node (x){};
      \path (1,-.6) node (v){};
      \path (0,-.6) node (y){};
      \path (1,0) node (u){};
      \path (1,.6) node (w){};
      \path (x)+(.2,-.13) coordinate (xx){};
      \path (u)+(-.2,-.13) coordinate (uu){};
      \path (v)+(-.2,.13) coordinate (vv){};
      \path (y)+(.2,.13) coordinate (yy){};

      \tikzstyle{every node}=[]
      \path (u) node[right=2pt]{$u$};
      \path (v) node[right=2pt]{$v$};
      \path (y) node[left=2pt]{$y$};
      \path (x) node[left=2pt]{$x$};
      \path (w) node[right=2pt]{$w$};
      \path (-.6,0) node{$V^+$};
      \path (1.6,0) node{$V^-$};

      \draw[bend left,gray!80] (u) edge (v);
      \draw[thick,red] (x)-- node[above=-2pt,pos=.84]{\Large $+$} (w);
      \draw[thick,green] (u)-- node[below=-3pt,pos=.16]{\scriptsize $-$} node[below=-5pt,pos=.84]{\Large $-$} (x);
      \draw[thick,green] (y)-- node[below=-4pt,pos=.18]{\Large $-$} node[below=-3pt,pos=.84]{\scriptsize $-$} (v);
      \draw[thick,gray!80,->,rounded corners=8pt] (xx) -- (uu) -- (vv) -- (yy);
    \end{tikzpicture}

    \caption{3-hop dismountability of a non $\{1,2\}$-hop dismountable clique.}
    \label{fig:lemma8}
\end{figure}

\begin{proof}
  We will prove the case $x,y\in V^+$ (see Figure~\ref{fig:lemma8}), the other case is symmetric. Assume by contraposition that $\lambda(ux)<\lambda(uv)<\lambda(vy)$ and let $w\in V^-$ be the unique node such that $n^+(w)=x$. Then $x$ can delegate its emissions to $y$ through the temporal path $xuvy$ and it can delegate its receptions to $w$ through the edge $xw$. Thus, $\G$ is 3-hop dismountable.
\end{proof}

\subsection{Non $k$-hop dismountable cliques}

\begin{theorem}\label{thm:threehopdismountable}
    Let $\G=((V,E),\lambda)$ be a temporal clique, if $\G$ is $k$-hop dismountable for $k>3$, then $\G$ is $\{1,2,3\}$-hop dismountable.
\end{theorem}

\begin{proof}
  By contradiction, let $k> 3$ be the smallest $k$ for which $\G$ is $k$-hop dismountable.
  Since $\G$ is not $\{1,2\}$-hop dismountable, then $V^+$ and $V^-$ form a partition of $V$. Let $v$ be a node that is $k$-hop dismountable
  and assume that $v\in V^+$, the other case is symmetrical. Since $v$ is in $V^+$, it can delegate its receptions to some node $w\in V^-$ such that $n^+(w)=v$.
  Now, let $u$ be the node that $v$ can reach through $e^-(u)$ using a temporal path $P$ of length $k$ along nodes $v,x_1,\dots x_{k-1},u$.
  For all $i=1,\dots ,k-1$, if $x_i$ belongs to $V^+$, then $x_i$ is $k'$-hop dismountable, with $k'<k$, thus all the $x_i$s must belong to $V^-$.
  Node $u$, on the other hand, can be either in $V^+$ or $V^-$, the argument will be the same in both cases (see~\cref{fig:k-hops}).
  Consider the node $x_{k-2}$. This node is in $V^-$, so there is a node $v'\in V^+$ such that $n^-(v')=x_{k-2}$.
  Either $\lambda(v'x_{k-2}) < \lambda(x_{k-2}x_{k-1})$ or $\lambda(v'x_{k-2}) > \lambda(x_{k-2}x_{k-1})$.
  If $\lambda(v'x_{k-2}) < \lambda(x_{k-2}x_{k-1})$, then $v'$ can reach $u$ through $e^-(u)$ in~$3$-hop,
  and since $v'\in V^+$, $v'$ can delegate its receptions to some node $w'\in V^-$ through $e^+(w')=v'w'$, thus $v'$ is $3$-hop dismountable.
  Now, if $\lambda(v'x_{k-2}) > \lambda(x_{k-2}x_{k-1})$, then we also have $\lambda(v'x_{k-2}) > \lambda(x_{k-3}x_{k-2})$
  (since $P$ is a temporal path), thus $v$ can reach $v'$ through $e^-(v')$ instead of reaching $u$ through $e^-(u)$,
  in $k-1$ hops, which contradicts the fact that $k$ is minimum.
\end{proof}

\begin{figure}[h]
    \centering
    \begin{tikzpicture}[xscale=1.4,yscale=1]
      \tikzstyle{every node}=[inner sep=1.6pt,circle,black,fill=darkgray]
      \path (0,0) node (u1){};
%      \path (0,1) node (u2){};
      \path (0,2) node (u3){};
      \path (0,3) node (u4){};
      \path (1,0) node (v1){};
%      \path (1,1) node (v2){};
      \path (1,2) node (v3){};
      \path (1,3) node (v4){};

      \tikzstyle{every node}=[]
      % \path (-1,2) node{\Large $V^+$};
      % \path (2,2) node{\Large $V^-$};
      \path (u1) node[left=2pt]{$v$};
      \path (u3) node[left=2pt]{$v'$};
      \path (u4) node[left=2pt]{$u$};
      \path (v3) node[right=2pt]{$x_{k-2}$};
      \path (v4) node[right=2pt]{$x_{k-1}$};

      \draw (u1) -- (v1);
      \draw (u3) -- node[above=-4pt,pos=.15]{$-$} (v3);
      \draw (u4) -- node[above=-4pt,pos=.15]{$-$} (v4);
      \draw (v3) -- (v4);
      \draw[snake it] (v1) -- (v3);
      \path (-.8,1.5) node {$V^+$};
      \path (1.8,1.5) node {$V^-$};
\end{tikzpicture}
\hspace{1.5cm}
\begin{tikzpicture}[xscale=1.4,yscale=1]
      \tikzstyle{every node}=[inner sep=1.6pt,circle,black,fill=darkgray]
      \path (0,0) node (u1){};
      \path (0,1) node (u2){};
      \path (0,2) node (u3){};
      \path (0,3) node (u4){};
      \path (1,0) node (v1){};
      \path (1,1) node (v2){};
      \path (1,2) node (v3){};
      \path (1,3) node (v4){};

      \tikzstyle{every node}=[]
      % \path (-1,2) node{\Large $V^+$};
      % \path (2,2) node{\Large $V^-$};
      \path (u1) node[left=2pt]{$v$};
      \path (u2) node[left=2pt]{$v'$};
      \path (v4) node[right=2pt]{$u$};
      \path (v2) node[right=2pt]{$x_{k-2}$};
      \path (v3) node[right=2pt]{$x_{k-1}$};

      \draw (u1) -- (v1);
      \draw (u3) -- (v3);
      \draw (u4) -- (v4);
      \draw (v2) -- (v3);
      \draw (u2) -- node[above=-4pt,pos=.15]{$-$} (v2);
      \draw (v3) -- node[right=-2pt,pos=.85]{$-$} (v4);
      \draw[snake it] (v1) -- (v2);
      \path (-.8,1.5) node {$V^+$};
      \path (1.8,1.5) node {$V^-$};
\end{tikzpicture}

    \caption{Illustration of the proof of Theorem~\ref{thm:threehopdismountable}.}
    \label{fig:k-hops}
\end{figure}

So, we have the interesting property that $k$-hop dismountability implies $\{1,2,3\}$-hop dismountability. Does this imply that if a clique is \emph{recursively} $k$-hop dismountable, then it admits a spanner of size at most $6n$? In fact, the implication does not work directly, because chosing a $\{1,2,3\}$-hop dismountable node at some step of the recursion, instead of another node with a larger $k$, may prevent \emph{recursive} dismountability: the order in which the nodes are dismounted seem to matter here. However, a stronger result actually holds for another reason, we prove in Section~\ref{sec:misc} that recursively $k$-hop dismountable cliques are pivotable, and thus admit a $2n-3$ spanner.

Back to non-recursive dismountability, a direct consequence of Theorem~\ref{thm:threehopdismountable} is that any \emph{minimal counter-example} to the existence of $6n$ spanners in temporal cliques must be non $k$-hop dismountable, as if it were, then it would be $\{1,2,3\}$-hop dismountable and could thus be reduced to a smaller instance by paying at most $6$ edges in the spanner. An interesting by-product of the proof of Theorem~\ref{thm:threehopdismountable} is that this is actually true also for $4n$.

\begin{corollary}\label{cor:4n-spanner}
  Any minimal counter-example to the existence of $4n$ spanners in temporal cliques must be non $k$-hop dismountable.
%   Let $\G=((V,E),\lambda)$ be a temporal clique that is $k$-hop dismountable and $S_X$ be a spanner among the nodes $X \subseteq V$. Then there exists a vertex $u$ such that $\G$ has a spanner of size at most $|S_{\G\setminus \{u\}}| + 4$.
\end{corollary}

\begin{proof}
  If $\G$ is $k$-hop dismountable, then by Theorem~\ref{thm:threehopdismountable}, there exists a node $u$ that is $\{1,2,3\}$-hop dismountable along two paths $P^-$ and $P^+$.
  % and $S = S_{\G\setminus \{u\}} \cup P^- \cup P^+$ is a spanner of $\G$.
  If $u$ is $\{1,2\}$-hop dismountable, then $|P^-| \leq 2$ and $|P^+| \leq 2$, thus at most $4$ edges must be added to the spanner before recursing. Otherwise, we know that either $u \in V^-$ or $u \in V^+$, which implies that one of the two paths consists of a single edge, and since $u$ is $3$-hop dismountable, the other path has length at most $3$, thus at most $4$ edges must be added to the spanner before recursing.
\end{proof}

\subsection{Summary of the characterized structure}

Here, we summarize the obtained properties and show that they characterize (non) dismountability in a necessary and sufficient manner. Because non $\{1,2\}$-hop dismountable cliques could be of independent interest, we proceed in two steps, summarizing the structure of non $\{1,2\}$-hop dismountable cliques, then extending it to non $\{1,2,3\}$-hop dismountable cliques.
Let $E^-=\{e^-(v) : v\in V\}$ (resp. $E^+=\{e^+(v) : v\in V\}$) be the set of edges that are earliest (latest) for at least one node. Then we have:
\begin{theorem}[Summary of non $\{1,2\}$-hop dismountability]\label{lemma:2-hop-summary}~\\
    $\G$ is not $\{1,2\}$-hop dismountable if and only if:
    \begin{enumerate}
        \item $V^-$ and $V^+$ are the same size and form a partition of $V$.
        \item Every edge between $V^-$ and $V^+$ is later than all adjacent edges in $E^-$ and earlier than all adjacent edges in $E^+$.
    \end{enumerate}
\end{theorem}

\begin{proof}
  $(\implies)$ Suppose that $\G$ is not $\{1,2\}$-hop dismountable. Property~$1$ is deduced from Corollary~\ref{cor:matching} and Lemma~\ref{lemma5}.
  Now, consider the edge $uv$ where $u\in V^-$ and $v\in V^+$. Because $u \in V^-$ and Property~$1$, it has only one edge in $E^+$ incident to it, namely its latest edge $e^+(u)$. Similarly,
  $v \in V^+$ has only its earliest edge $e^-(v)$ that is part of $E^-$.

  If $uv$ was earlier than an adjacent edge $f \in E^-$, then $f$ must be incident to $u$ (as the only one incident to $v$ in $E^-$ is $e^-(v)$).
  Moreover, since $f$ is later than $uv$ and $f \neq e^-(u)$, there is a 2-hop temporal path from $v$ to the earliest edge of another node, making $v$ 2-hop dismountable, leading to a contradiction.
  Hence $uv$ must be later than any adjacent edge from $E^-$ and a similar argument shows that it is earlier than any adjacent edge from $E^+$.

  $(\impliedby)$  Assume that the above properties hold. Suppose by contradiction that there exist a node $v$ that is dismountable.
  Since $V^-$ and $V^+$ form a partition of $V$, suppose that $v \in V^+$ (the case $v \in V^-$ being symmetrical).

  Observe that $V^-\cap V^+=\emptyset$ is equivalent to saying that $\G$ is not 1-hop dismountable, meaning that $v$ must be 2-hop dismountable.
  Since $v\in V^+$, $v$ is 2-hop dismountable if there exists a 2-hop temporal path $P$ on nodes $v,x,w$ where $x=n^-(w)$. By definition of $V^-$, this means that $x\in V^-$.

  Finally, we have that $vx$ is an edge between $V^-$ and $V^+$ and $xw\in E^-$.
  Thus, by Property~$2$ we have that $\lambda(vx)>\lambda(xw)$, hence $P$ is not a temporal path and $v$ is not 2-hop dismountable.
\end{proof}

Let $M^-=\{e^-(v) : v\in V^+\}$ (resp. $M^+=\{e^+(v) : v\in V^-\}$) be the set of edges between $V^-$ and $V^+$ that are in $E^-$ (resp. in $E^+$). By Corollary~\ref{cor:matching}, these two sets form perfect matchings when $\G$ is not $2$-hop dismountable.

\begin{theorem}[Summary of non $k$-hop dismountability]\label{lemma:non-1,2,3-hop}~\\
  $\G$ is not $k$-hop dismountable if and only if:
  \begin{enumerate}
      \item $V^-$ and $V^+$ are the same size and form a partition of $V$.
      \item Every edge between $V^-$ and $V^+$ is later than all adjacent edges in $E^-$ and earlier than all adjacent edges in $E^+$.
      \item For every edge $e$ within the part $V^-$ (resp. $V^+$), the label of $e$ cannot be between the labels of the two incident edges of $M^-$ (resp. $M^+$).
  \end{enumerate}
\end{theorem}

\begin{proof}
  By Theorem~\ref{lemma:2-hop-summary}, the first two statements are equivalent to $\G$ not being $\{1,2\}$-hop dismountable.
  By Theorem~\ref{thm:threehopdismountable}, a $k$-hop dismountable graph is $3$-hop dismountable, hence we just need to prove that adding the third statement is equivalent to non $3$-hop dismountability.

  ($\implies$) The third statement is implied by Lemma~\ref{lemma10}.

  ($\impliedby$) Assume that a node $v \in V^+$ is $3$-hop dismountable and no node is $2$-hop dismountable.
  This implies that there is a $3$-hop temporal path $P^-$ from $v$ to another node $u$ through $e^-(u)$.
  Let $P^-$ be on nodes $v,x_1,x_2,u$, and observe that $x_1$ and $x_2$ cannot be in $V^+$, as otherwise they are $2$-hop dismountable and $1$-dismountable respectively.
  Thus they are both in $V^-$.

  We may assume that $vx_1 = e^-(v)$ as otherwise we can substitute $v$ with the node $y_1$ such that $y_1x_1 = e^-(y_1)$ and Property~$2$ guarantees that $\lambda(y_1x_1)<\lambda(vx_1)$.
  Moreover, we may assume that $u\in V^+$ as otherwise we can substitute $u$ with the node $y_2\in V^+$ such that $n^-(y_2)=x_2$ and Property~$2$ once again guarantees that $\lambda(x_2u)<\lambda(x_2y_2)$.
  Summing up, we have that $vx_1\in M^-$, $x_2u\in M^-$ and the edge $x_1x_2$  is between two incident edges of $M^-$ in term of labels, contradicting Property~$3$.
\end{proof}

% \begin{lemma}[No sibling leaves] \label{lemma6}
%   The internal nodes in $E^-$ have at most one leaf neighbor in $E^-$. (The same holds for $E^+$.)
%   \todo{missing definition of $E^+$ and $E^-$}
% \end{lemma}

% \begin{proof}
%   We prove the property for $E^-$, the argument for $E^+$ being symmetric. By definition, the internal nodes of $E^-$ are the nodes in $V^-$. Since $V^0$ is empty, the leaves of $E^-$ are exactly the nodes of $V^+$, each of which must have a distinct minimum neighbor (\cref{lem:simple-2-hop}).
% \end{proof}

\section{Dismountability is all you need for sparse spanners}
\label{sec:spanners}

In the previous section, we gave a complete characterization of temporal cliques that are not dismountable. In particular, we showed that if a clique is $k$-hop dismountable for any $k$, then it must also be either $1$-hop, $2$-hop, or $3$-hop dismountable. Among other implications, our results imply that any minimal counter-example to the existence of a $4n$ spanner should possess all the structure described in Theorem~\ref{lemma:non-1,2,3-hop}.

If one does not care about the specific constant, but only about whether temporal cliques admit $O(n)$ spanners, then it is actually sufficient to exclude only $1$-hop and $2$-hop dismountability. Indeed, non $\{1,2\}$-hop dismountable cliques already possess the required structure for reducing the problem to a bipartite version, and the findings of~\cite{Hasso} imply that the upper bound on the size of a spanner in this setting may differ only by a constant factor from that of temporal cliques. % Furthermore, the structure added by excluding $3$-hop dismountability does no affect the possible labelings of edges between $V^+$ and $V^-$.

Let us recall the bipartite version of the problem. The definition of the bi-clique itself follows from Theorem~\ref{lemma:2-hop-summary}. Namely,

\begin{corollary}
  \label{cor:reduction}
  Let $\G$ be a non $\{1,2\}$-hop dismountable graph and let $\G'$ be a temporal bi-clique obtained by restricting $\G$ to the edges between $V^-$ and $V^+$. Then, $\G'$ has the following properties:
  \begin{enumerate}
  \item $V^-$ and $V^+$ have the same size and form a partition of $V$.
  \item The set $M^-:=\{e^-(v) : v\in V^+\}=\{e^-(v) : v\in V^-\}$ is a perfect matching.
  \item The set $M^+:=\{e^+(v) : v\in V^-\}=\{e^+(v) : v\in V^+\}$ is a perfect matching.
  \end{enumerate}
\end{corollary}

Corollary~\ref{cor:reduction} is illustrated on Figure~\ref{fig:bipartite}.

     \begin{figure}[h]
    \centering
    \begin{tikzpicture}[xscale=2,yscale=1.5]
      \tikzstyle{every node}=[inner sep=1.4pt,circle,black,fill=darkgray]
      \path (0,0) node (u1){};
      \path (0,1) node (u2){};
      \path (0,2) node (u3){};
      \path (0,3) node (u4){};
      \path (1,0) node (v1){};
      \path (1,1) node (v2){};
      \path (1,2) node (v3){};
      \path (1,3) node (v4){};

      \tikzstyle{every node}=[]
      \path (-.6,2) node{\Large $V^+$};
      \path (-.7,.5) node{\LARGE $\G$};
      \path (1.7,2) node{\Large $V^-$};

      \draw[lightgray] (u1) -- (v3);
      \draw[lightgray] (u1) -- (v4);
      \draw[lightgray] (u2) -- (v2);
      \draw[lightgray] (u2) -- (v4);
      \draw[lightgray] (u3) -- (v1);
      \draw[lightgray] (u3) -- (v3);
      \draw[lightgray] (u4) -- (v1);
      \draw[lightgray] (u4) -- (v2);
      \draw[thick,red] (u1) -- node[below=-2pt,pos=.1]{\tiny $+$} node[below=-2pt,pos=.92]{\Large $+$} (v1);
      \draw[thick,green] (u1) -- node[below=-2pt,pos=.18]{\Large $-$} node[above=-2pt,pos=.88]{\tiny $-$} (v2);
      \draw[thick,green] (u2) -- node[above=-2pt,pos=.18]{\Large $-$} node[below=-2pt,pos=.88]{\tiny $-$} (v1);
      \draw[thick,red] (u2) -- node[below=-2pt,pos=.1]{\tiny $+$} node[below=-2pt,pos=.92]{\Large $+$} (v3);
      \draw[thick,red] (u3) -- node[above=-2pt,pos=.1]{\tiny $+$} node[above=-2pt,pos=.92]{\Large $+$} (v2);
      \draw[thick,green] (u3) -- node[below=-2pt,pos=.18]{\Large $-$} node[above=-2pt,pos=.88]{\tiny $-$} (v4);
      \draw[thick,green] (u4) -- node[above=-2pt,pos=.18]{\Large $-$} node[below=-2pt,pos=.88]{\tiny $-$} (v3);
      \draw[thick,red] (u4) -- node[above=-2pt,pos=.1]{\tiny $+$} node[above=-2pt,pos=.92]{\Large $+$} (v4);

      \tikzstyle{every edge}=[draw,bend left]
      \draw[lightgray,bend left=15] (u1) edge (u2);
      \draw[lightgray] (u1) edge (u3);
      \draw[thick,red] (u1) edge node[below left=-4pt,pos=.04]{\Large $+$} (u4);
      \draw[thick,red,bend left=15] (u2) edge node[left=-2pt,pos=.13]{\Large $+$} node[left=-2pt,pos=.87]{\Large $+$} (u3);
      \draw[lightgray] (u2) edge (u4);
      \draw[thick,red,bend left=15] (u3) edge node[right=-2pt,pos=.85]{\Large $+$} (u4);

      \tikzstyle{every edge}=[draw,bend right]
      \draw[thick,green,bend left=15] (v1) edge node[right=-2pt,pos=.18]{\Large $-$} node[right=-2pt,pos=.88]{\Large $-$} (v2);
      \draw[lightgray] (v1) edge (v3);
      \draw[lightgray] (v1) edge (v4);
      \draw[lightgray,bend left=15] (v2) edge (v3);
      \draw[lightgray] (v2) edge (v4);
      \draw[thick,green,bend left=15] (v3) edge node[right=-2pt,pos=.18]{\Large $-$} node[right=-2pt,pos=.88]{\Large $-$} (v4);

    \end{tikzpicture}
    ~~~~~~~~~~~
    \begin{tikzpicture}[xscale=2,yscale=1.5]
      \tikzstyle{every node}=[inner sep=1.4pt,circle,black,fill=darkgray]
      \path (0,0) node (u1){};
      \path (0,1) node (u2){};
      \path (0,2) node (u3){};
      \path (0,3) node (u4){};
      \path (1,0) node (v1){};
      \path (1,1) node (v2){};
      \path (1,2) node (v3){};
      \path (1,3) node (v4){};

      \tikzstyle{every node}=[]
      \path (-.4,2) node{\Large $V^+$};
      \path (-.5,.5) node{\LARGE $\G'$};
      \path (1.6,2) node{\Large $V^-$};

      \draw[lightgray] (u1) -- (v3);
      \draw[lightgray] (u1) -- (v4);
      \draw[lightgray] (u2) -- (v2);
      \draw[lightgray] (u2) -- (v4);
      \draw[lightgray] (u3) -- (v1);
      \draw[lightgray] (u3) -- (v3);
      \draw[lightgray] (u4) -- (v1);
      \draw[lightgray] (u4) -- (v2);
      \draw[thick,red] (u1) -- node[below=-2pt,pos=.1]{\small $+$} node[below=-2pt,pos=.92]{\small $+$} (v1);
      \draw[thick,green] (u1) -- node[below=-2pt,pos=.18]{\small $-$} node[above=-2pt,pos=.88]{\small $-$} (v2);
      \draw[thick,green] (u2) -- node[above=-2pt,pos=.18]{\small $-$} node[below=-2pt,pos=.88]{\small $-$} (v1);
      \draw[thick,red] (u2) -- node[below=-2pt,pos=.1]{\small $+$} node[below=-2pt,pos=.92]{\small $+$} (v3);
      \draw[thick,red] (u3) -- node[above=-2pt,pos=.1]{\small $+$} node[above=-2pt,pos=.92]{\small $+$} (v2);
      \draw[thick,green] (u3) -- node[below=-2pt,pos=.18]{\small $-$} node[above=-2pt,pos=.88]{\small $-$} (v4);
      \draw[thick,green] (u4) -- node[above=-2pt,pos=.18]{\small $-$} node[below=-2pt,pos=.88]{\small $-$} (v3);
      \draw[thick,red] (u4) -- node[above=-2pt,pos=.1]{\small $+$} node[above=-2pt,pos=.92]{\small $+$} (v4);
    \end{tikzpicture}

    \caption{A possible configuration of the earliest and latest edges of a non $\{1,2\}$-hop dismountable clique and the corresponding bi-clique.}
    \label{fig:bipartite}
\end{figure}

First, observe that $\G'$ is always temporally connected. Indeed, every node can reach the nodes of the opposite part using a single edge, and can reach the other nodes of its own part using $2$-hop temporal paths that start with the earliest edges. Furthermore, since $\G'$ is based on the same node set as $\G$, any temporal spanner of $\G'$ is a temporal spanner of~$\G$.

Last but not least, thanks to the edges of $M^-$ and $M^+$, it is sufficient to focus on finding a spanner that achieves reachability \emph{from} $V^+$ \emph{to} $V^-$ (or alternatively, from $V^-$ to $V^+$). Indeed, whenever a node of $V^+$ can reach all the nodes of $V^-$, one can extend the corresponding path(s) from each node of $V^-$ by using an edge of $M^+$, collectively reaching all of $V^+$ (if this edge is already the last one of the temporal path, no such extension is needed). Similarly, the nodes of $V^-$ can reach each other by prefixing the same paths with edges of $M^-$. This feature facilitate the computation of a spanner in the bipartite graph.

\subsection{One-sided dismountability and recursion}
\label{sec:hasso}

Let $\G$ be a temporal bi-clique as defined above. For completeness, this section recalls the strategy used in~\cite{Hasso} to simplify the second part of the original proof and recover the $O(n \log n)$ result more elegantly.
% discuss how a similar strategy could inspire a more general approach for proving (or refuting) the existence of $O(n)$ spanners, which we hope can be useful for future studies.
Let us denote here by $n$ the size of each part $V^+$ and $V^-$ (total size $2n$).
The first component of the strategy is to divide the work into two parts, by splitting $V^+$ into two sets of nodes $X_1$ and $X_2$, each of size $n/2$, and recurse in two subproblems where the goal is to compute a spanner from $X_1$ to all of $V^-$, and a spanner from $X_2$ to all of $V^-$ (see Figure~\ref{fig:hasso}). The union of these two spanners ensures that every node of $V^+$ can reach all of $V^-$, which together with the matchings $M^-$ and $M^+$, guarantees that the resulting graph is temporally connected (as discussed above). The major observation is then that a one-sided version of dismountability can be used to reduce each of the subproblems until both parts have equal size. Indeed, as long as the target part (call it $T$, which is $V^-$ in the first step) is larger than the source part (call it $S$, which is $X_1$ or $X_2$ in the first step), then at least two nodes in $T$ have their latest neighbor coinciding in $S$. Thus, one of these nodes can delegate its receptions (relative to $S$) to the other, and be dismounted from the instance at the cost of selecting the corresponding two edges in the spanner. This is indeed sufficient because reachability \emph{from} $S$ \emph{to} $T$ is sufficient. One can then split the work, and recurse again. 

\begin{figure}[h]
\begin{tikzpicture}[yscale=.12,xscale=.95]
  \tikzstyle{every node}=[defnode,inner sep=.8pt]
  \path (0,26) node {};
  \path (0,25) node {};
  \path (0,24) node {};
  \path (0,23) node {};
  \path (0,22) node {};
  \path (0,21) node {};
  \path (0,20) node {};
  \path (0,19) node {};
  \path (0,18) node {};
  \path (0,17) node {};
  \path (0,16) node {};
  \path (0,15) node {};
  \path (0,14) node {};
  \path (0,13) node {};
  \path (0,12) node {};
  \path (0,11) node {};

  \path (1,26) node {};
  \path (1,25) node {};
  \path (1,24) node {};
  \path (1,23) node {};
  \path (1,22) node {};
  \path (1,21) node {};
  \path (1,20) node {};
  \path (1,19) node {};
  \path (1,18) node {};
  \path (1,17) node {};
  \path (1,16) node {};
  \path (1,15) node {};
  \path (1,14) node {};
  \path (1,13) node {};
  \path (1,12) node {};
  \path (1,11) node {};

  \tikzstyle{every node}=[defnode,inner sep=.8pt]
  \path (3,37) node {};
  \path (3,36) node {};
  \path (3,35) node {};
  \path (3,34) node {};
  \path (3,33) node {};
  \path (3,32) node {};
  \path (3,31) node {};
  \path (3,30) node {};

  \path (4,37) node {};
  \path (4,36) node {};
  \path (4,35) node {};
  \path (4,34) node {};
  \path (4,33) node {};
  \path (4,32) node {};
  \path (4,31) node {};
  \path (4,30) node {};
  \path (4,29) node {};
  \path (4,28) node {};
  \path (4,27) node {};
  \path (4,26) node {};
  \path (4,25) node {};
  \path (4,24) node {};
  \path (4,23) node {};
  \path (4,22) node {};

  \tikzstyle{every node}=[defnode,inner sep=.8pt]
  \path (3,7) node {};
  \path (3,6) node {};
  \path (3,5) node {};
  \path (3,4) node {};
  \path (3,3) node {};
  \path (3,2) node {};
  \path (3,1) node {};
  \path (3,0) node {};

  \path (4,15) node {};
  \path (4,14) node {};
  \path (4,13) node {};
  \path (4,12) node {};
  \path (4,11) node {};
  \path (4,10) node {};
  \path (4,9) node {};
  \path (4,8) node {};
  \path (4,7) node {};
  \path (4,6) node {};
  \path (4,5) node {};
  \path (4,4) node {};
  \path (4,3) node {};
  \path (4,2) node {};
  \path (4,1) node {};
  \path (4,0) node {};

  \tikzstyle{every node}=[defnode,inner sep=.8pt]
  \path (6,37) node {};
  \path (6,36) node {};
  \path (6,35) node {};
  \path (6,34) node {};
  \path (6,33) node {};
  \path (6,32) node {};
  \path (6,31) node {};
  \path (6,30) node {};

  \path (7,37) node {};
  \path (7,36) node {};
  \path (7,35) node {};
  \path (7,34) node {};
  \path (7,33) node {};
  \path (7,32) node {};
  \path (7,31) node {};
  \path (7,30) node {};

  \tikzstyle{every node}=[defnode,inner sep=.8pt]
  \path (6,7) node {};
  \path (6,6) node {};
  \path (6,5) node {};
  \path (6,4) node {};
  \path (6,3) node {};
  \path (6,2) node {};
  \path (6,1) node {};
  \path (6,0) node {};

  \path (7,7) node {};
  \path (7,6) node {};
  \path (7,5) node {};
  \path (7,4) node {};
  \path (7,3) node {};
  \path (7,2) node {};
  \path (7,1) node {};
  \path (7,0) node {};

  \tikzstyle{every node}=[defnode,inner sep=.8pt]
  \path (9,37) node {};
  \path (9,36) node {};
  \path (9,35) node {};
  \path (9,34) node {};

  \path (10,37) node {};
  \path (10,36) node {};
  \path (10,35) node {};
  \path (10,34) node {};
  \path (10,33) node {};
  \path (10,32) node {};
  \path (10,31) node {};
  \path (10,30) node {};

  \tikzstyle{every node}=[defnode,inner sep=.8pt]
  \path (9,23) node {};
  \path (9,22) node {};
  \path (9,21) node {};
  \path (9,20) node {};

  \path (10,27) node {};
  \path (10,26) node {};
  \path (10,25) node {};
  \path (10,24) node {};
  \path (10,23) node {};
  \path (10,22) node {};
  \path (10,21) node {};
  \path (10,20) node {};

  \tikzstyle{every node}=[defnode,inner sep=.8pt]
  \path (9,17) node {};
  \path (9,16) node {};
  \path (9,15) node {};
  \path (9,14) node {};

  \path (10,17) node {};
  \path (10,16) node {};
  \path (10,15) node {};
  \path (10,14) node {};
  \path (10,13) node {};
  \path (10,12) node {};
  \path (10,11) node {};
  \path (10,10) node {};

  \tikzstyle{every node}=[defnode,inner sep=.8pt]
  \path (9,3) node {};
  \path (9,2) node {};
  \path (9,1) node {};
  \path (9,0) node {};

  \path (10,7) node {};
  \path (10,6) node {};
  \path (10,5) node {};
  \path (10,4) node {};
  \path (10,3) node {};
  \path (10,2) node {};
  \path (10,1) node {};
  \path (10,0) node {};

  \tikzstyle{every node}=[defnode,inner sep=.8pt]
  \path (12,37) node {};
  \path (12,36) node {};
  \path (12,35) node {};
  \path (12,34) node {};

  \path (13,37) node {};
  \path (13,36) node {};
  \path (13,35) node {};
  \path (13,34) node {};

  \tikzstyle{every node}=[defnode,inner sep=.8pt]
  \path (12,26) node {};
  \path (12,25) node {};
  \path (12,24) node {};
  \path (12,23) node {};

  \path (13,26) node {};
  \path (13,25) node {};
  \path (13,24) node {};
  \path (13,23) node {};

  \tikzstyle{every node}=[defnode,inner sep=.8pt]
  \path (12,14) node {};
  \path (12,13) node {};
  \path (12,12) node {};
  \path (12,11) node {};

  \path (13,14) node {};
  \path (13,13) node {};
  \path (13,12) node {};
  \path (13,11) node {};

  \tikzstyle{every node}=[defnode,inner sep=.8pt]
  \path (12,3) node {};
  \path (12,2) node {};
  \path (12,1) node {};
  \path (12,0) node {};

  \path (13,3) node {};
  \path (13,2) node {};
  \path (13,1) node {};
  \path (13,0) node {};

  \tikzstyle{every node}=[font=\footnotesize]
  \path (2,39) node {Split};
  \path (5,39) node {Dismount};
  \path (8,39) node {Split};
  \path (11,39) node {Dismount};

  \path (0.5,18) node {$\leadsto$};
  \path (3.5,30) node {$\leadsto$};
  \path (3.5,6) node {$\leadsto$};
  \path (6.5,32) node {$\leadsto$};
  \path (6.5,4) node {$\leadsto$};
  \path (9.5,34) node {$\leadsto$};
  \path (9.5,23) node {$\leadsto$};
  \path (9.5,14) node {$\leadsto$};
  \path (9.5,2) node {$\leadsto$};
  \path (12.5,35) node {$\leadsto$};
  \path (12.5,24) node {$\leadsto$};
  \path (12.5,12) node {$\leadsto$};
  \path (12.5,1) node {$\leadsto$};

  \tikzstyle{every path}=[thick,shorten <=5pt,shorten >=5pt]
  \draw (1.5,19) edge[->] (2.5,26);
  \draw (1.5,17) edge[->] (2.5,10);
  \draw (4.5,32) edge[->] (5.5,32);
  \draw (7.5,32) edge[->] (8.5,32);
  \draw (7.5,31) edge[->] (8.5,25);
  \draw (4.5,5) edge[->] (5.5,5);
  \draw (7.5,6) edge[->] (8.5,12);
  \draw (7.5,5) edge[->] (8.5,5);
  \draw (10.5,35) edge[->] (11.5,35);
  \draw (10.5,24) edge[->] (11.5,24);
  \draw (10.5,12) edge[->] (11.5,12);
  \draw (10.5,2) edge[->] (11.5,2);

  \draw (2.5,18.5) edge[ultra thick,lightgray] (13.3,18.5);
  \draw (13.08,18.5) edge[ultra thick,lightgray,dash pattern=on 3pt off 3pt] (13.8,18.5);

  \draw (8.5,28.5) edge[ultra thick,lightgray] (13.3,28.5);
  \draw (13.08,28.5) edge[ultra thick,lightgray,dash pattern=on 3pt off 3pt] (13.8,28.5);

  \draw (8.5,8.5) edge[ultra thick,lightgray] (13.3,8.5);
  \draw (13.08,8.5) edge[ultra thick,lightgray,dash pattern=on 3pt off 3pt] (13.8,8.5);
\end{tikzpicture}
\caption{\label{fig:hasso}One-sided dismountability and recursion (strategy from~\cite{Hasso}).}
\end{figure}

In each step, dismountability may work beyond balancing the two parts, but in any case, it will work \emph{at least} to that point. In each step, the number of selected edges when dismounting is $O(|T|)$, which yields the following recurrence on the size of a spanner: $size(n) \le 2\cdot size(n/2) + O(n) = O(n\log n)$ edges (by the master theorem).

\subsection{Epilogue}

The above strategy from~\cite{Hasso} completes the overall proof that temporal cliques admit spanners of size $O(n \log n)$, using dismountability only. This proof is constructive. The corresponding algorithm is polynomial and is summarized in Algorithm~\ref{algo}.

\begin{algorithm}
  \begin{algorithmic}
    \Require A temporal clique $\G=(V,E,\lambda)$\smallskip
    \State \hspace*{-10pt}\texttt{Spanner}($\G$):
    \If {$\exists$ $1$-hop or $2$-hop dismountable node $v\in V$}
    \State $E' \gets $ edges used in the dismountability paths \Comment{Note that $|E'|\le 4$}
    \State \Return $E'~\cup$ \texttt{Spanner}($\G \setminus \{v\}$)
    \Else
    \State $\G' \gets \G$ restricted to the edges between $V^+$ and $V^-$
    \State \Return \texttt{Bipartite\_Spanner}($\G'$) \Comment{At most $O(n \log n)$ edges}
    \EndIf
\end{algorithmic}
\caption{\label{algo}Recursive algorithm for computing a $O(n \log n)$ spanner in a temporal clique. The \texttt{Bipartite\_Spanner} function corresponds to the part of the algorithm from~\cite{Hasso} outlined in Section~\ref{sec:hasso}.}
\end{algorithm}

The search for $1$-hop or $2$-hop dismountable node can be done as follows. First, mark all the earliest and latest edge at every node. Then, for every node $v$, compute the subset of nodes that $v$ can reach in at most two hops, arriving through the earliest edge at these nodes (memorize at least one such path for each of these nodes), and compute the subset of nodes that can reach $v$ in at most two hops by departing through their latest edge (and memorize at least one such path). If both subsets are non-empty, then $v$ is $\{1,2\}$-hop dismountable through these paths. It is not hard to see that each of these steps runs in polynomial time, and so does the \texttt{Bipartite\_Spanner} algorithm (Corollary 4.8 in~\cite{Hasso}).

\section{Dismountability, Pivotability, and Full-range graphs}
\label{sec:misc}

In this section, we provide a collection of further results and observations related to dismountability. In particular, we discuss some relations between dismountability and pivotability, and present a class of temporal graphs that are both dismountable and pivotable.

For simplicity, let us recall the definition of the concept of a pivot edge in the setting of simple and proper graphs. This concept can be adapted to more general settings, as well as to a version with pivot nodes or larger pivotal structures (irrelevant in the present paper).

\begin{definition}[Pivotability]
  Let $\G=(V,E,\lambda)$ be a simple and proper temporal graph in \textsf{TC}, not necessarily a clique. A pivot edge is an edge $e\in E$ such that for every node $v\in V$, there exists a temporal path from $v$ to $e$ (the path ends by traversing $e$ at time $\lambda(e)$), and there exists a temporal path from $e$ to $v$ (the path starts by traversing $e$ at time $\lambda(e)$). A temporal graph that admits a pivot edge is said to be pivotable.
\end{definition}

A simple observation in~\cite{BFJ03} implies that if a pivot edge $e$ exists, then there must also exist a temporal \emph{in-tree} from all the nodes to $e$, and a temporal \emph{out-tree} from $e$ to all the nodes, the union of which gives a $O(n)$ spanner (indeed, a $2n-3$ spanner).
The following theorem establishes a relation between dismountability and pivotability.

\begin{theorem}
Let $\G=(V,E,\lambda)$ be a temporal clique. If $\G$ is recursively $k$-hop dismountable, then $\G$ is pivotable.
\end{theorem}

\begin{proof}
Let us denote with $\langle u_i,P^-_i,P^+_i \rangle_{i\in [n-2]}$ the sequence of nodes and temporal paths witnessing that $\G$ is recursively $k$-hop dismountable; i.e. $P^-_i$ and $P^+_i$ are temporal paths in $\G \setminus \{v_j\ \mid j\in [i]\}$
from $u_i$ to some node $v$ arriving through $e^-(v)$, and from some node $w$ to $u_i$ starting with $e^+(w)$, respectively.
Moreover, let $x$ and $y$ be the last two remaining nodes after $\G$ has been recursively $k$-hop dismounted, and let $t=\lambda(xy)$.
We will prove by induction that $xy$ is a pivot edge in $\G$; namely, every $u_i$ can reach either $x$ or $y$ before time $\lambda(xy)$, and it can be reached by either $x$ or $y$ after time $\lambda(xy)$.

The base case is when $\G$ contains only 3 nodes, namely $u,x,y$.
Since $P^-_1$ is a temporal path from $u$ to $x$ or $y$ arriving through $e^-(x)$ or $e^-(y)$, we have that $u$ can reach $xy$.
As $P^+_1$ is a temporal path from $x$ or $y$ to $u$ starting with $e^+(x)$ or $e^+(y)$, we have that $u$ can be reached by $xy$.

Induction: Assume that all recursively $k$-hop dismountable graphs of size $n-1$ are pivotable with respect to an edge $xy$, and let $\G$ be a recursively $k$-hop dismountable graph of size $n$.
Let $\langle u_i,P^-_i,P^+_i \rangle_{i\in [n-2]}$ be the sequence witnessing that $\G$ is recursively $k$-hop dismountable.
By the induction hypothesis, in $\G\setminus \{u_1\}$ every node $u_i, i\in \{2,3,\dots,n-2\}$ can reach and be reached by $xy$.
We will prove that $u_1$ can reach either $x$ or $y$ before time $t$ and can be reached by either $x$ or $y$ after time $t$.
By definition of $P^-_1$ node $u_1$ can reach some other node $v$ arriving through the edge $e^-(v)$.
If $v$ is $x$ or $y$ then we have that $\lambda(e^-(v))\leq t$ and therefore $u_1$ can reach $xy$.
Otherwise, $v \in \{u_i \mid i\in \{2,3,\dots,n-2\}\}$. In this case, there must exist a temporal path in $\G\setminus \{u_1\}$ from $v$ to $xy$ arriving before time $t$.
As $u_1$ can reach $v$ through $e^-(v)$ using $P^-_1$, we can compose $P^-_1$ to the path from $v$ to $xy$ given by the induction hypothesis obtaining a temporal path from $u_1$ to $xy$ arriving before time $t$.
The proof that $xy$ can reach $u_1$ after time $t$ is symmetrical.
\end{proof}

\begin{corollary}
  Cliques that are recursively $k$-hop dismountable admit $2n-3$ spanners.
\end{corollary}

\subsection{Full-range temporal graphs}

Let us start by defining the \emph{time-compressed} version of a simple and proper temporal graph (not necessarily a clique) as follows.

\begin{definition}[compressed labeling]
  Let $\G=(V,E,\lambda)$ be a simple and proper temporal graph of lifetime $L$ (i.e. $\lambda$ takes its values in the range $[1,L]$). The labeling $\lambda$ is said to be \emph{time-compressed} (or simply compressed) if for all edge $e$, either $e$ is adjacent to an edge $e'$ such that $\lambda(e')=\lambda(e)-1$, or $\lambda(e)=1$.
\end{definition}

In other words, every label of the graph is as small as it can be while preserving the local ordering among edges of $\G$. Figure~\ref{fig:isomorphic-graphs} (left and middle) shows an example of temporal graph and its compressed version. (Again, this concept could be defined in more general settings than simple and proper graphs, bringing unnecessary complications to the present paper.)

\begin{figure}[ht]
\centering
  \begin{tikzpicture}[scale=1.4]
    \tikzstyle{every node}=[defnode]
    \path (0,0) node (a) {};
    \path (1,0) node (b) {};
    \path (0,1) node (c) {};
    \path (1,1) node (d) {};
    \path (.5,2) node (e) {};
    \tikzstyle{every node}=[inner sep=2pt]
    \draw (a) -- node[below] {5}(b);
    \draw (a) -- node[left] {3}(c);
    \draw (b) -- node[right] {4}(d);
    \draw (c) -- node[above=-1pt] {1}(d);
    \draw (c) -- node[above left] {2}(e);
    \draw (d) -- node[above right] {17}(e);
    \tikzstyle{every node}=[]
    \path (2,.8) node(simeq){\huge $\simeq$};
    \path (.5,-.6) node (g1) {$\G_1$};
  \end{tikzpicture}
  ~~~
  \begin{tikzpicture}[scale=1.4]
    \tikzstyle{every node}=[defnode]
    \path (0,0) node (a) {};
    \path (1,0) node (b) {};
    \path (0,1) node (c) {};
    \path (1,1) node (d) {};
    \path (.5,2) node (e) {};
    \tikzstyle{every node}=[inner sep=2pt]
    \draw (a) -- node[below] {4}(b);
    \draw (a) -- node[left] {3}(c);
    \draw (b) -- node[right] {2}(d);
    \draw (c) -- node[above=-1pt] {1}(d);
    \draw (c) -- node[above left] {2}(e);
    \draw (d) -- node[above right] {3}(e);
    \path (.5,-.6) node (g1) {$\G_{2}$};
  \end{tikzpicture}
  ~~~~~~~~~~~~~~~~~~~~~~~~~~~~
  \begin{tikzpicture}[scale=1.4]
    \tikzstyle{every node}=[defnode]
    \path (0,0) node (a) {};
    \path (1,0) node (b) {};
    \path (0,1) node (c) {};
    \path (1,1) node (d) {};
    \path (.5,2) node (e) {};
    \tikzstyle{every node}=[inner sep=2pt]
    \draw (a) -- node[below] {1}(b);
    \draw (a) -- node[left] {2}(c);
    \draw (b) -- node[right] {6}(d);
    \draw (c) -- node[above=-1pt] {3}(d);
    \draw (c) -- node[above left] {4}(e);
    \draw (d) -- node[above right] {5}(e);
    \path (.5,-.6) node (g1) {$\G_{3}$};
  \end{tikzpicture}
  \caption{Left and middle: A temporal graph and its time-compressed version. Right: a full-range temporal graph.}
  \label{fig:isomorphic-graphs}
\end{figure}

For all matters concerning the reachability relation alone, temporal graphs and their time-compressed versions are equivalent because they have the same \emph{local} ordering of the edges. Analysing the compressed version directly may be convenient.

\begin{definition}[Full-range graph]
  Let $\G=(V,E,\lambda)$ be the compressed version of a simple and proper temporal graph (not necessarily a clique). Then $\G$ is \emph{full-range} if its lifetime $L$ is equal to $|E|$.
\end{definition}

An example of full-range graph is shown in Figure~\ref{fig:isomorphic-graphs} (right). Observe that being full-range implies, among other things, that every label is used exactly once.
In the following, we prove that full-range graphs in \textsf{TC} are pivotable and thus admit $O(n)$ spanners. Moreover, full-range \emph{cliques} are $\{1,2,3\}$-hop dismountable. The following lemma is instrumental to both results.

\begin{lemma}\label{lemma-fullrange}
  Let $\G=(V,E,\lambda)$ be a temporal graph (not necessarily a clique) and let $e$ and $f$ be two edges such that $\lambda(e) < \lambda(f)$. If $\G$ is full-range, then the endpoints of $e$ can reach the endpoints of $f$ by time $\lambda(f)$.
\end{lemma}

\begin{proof}Let $t=\lambda(e)$ and $t'=\lambda(f)$.
  By induction on $k=t'-t$:

  If $t' = t+1$, then $e,f$ are adjacent and can be composed into a temporal path, so the endpoints of $e$ can reach the endpoints of $f$ at time $t+1$. This closes the base case.

  Induction: Consider $e_{t+k}$ and $e_{t+k-1}$, the edges with label $t+k$ and $t+k-1$ respectively.
  Since the graph is full-range, these two edges share a common node $w$.  By induction hypothesis, there is a temporal path $P$ from $e$ to $e_{t+k-1}$.
  Either $P$ ends in $w$ and can be composed with $e_{t+k}$ directly, or $w$ is the penultimate node of $P$,
  in which case a temporal path $P'$ can be constructed by keeping all the edges of $P$ but $e_{t+k-1}$ and by adding $e_{t+k}$.
\end{proof}

\begin{theorem}
  \label{full-range-pivotable}
  Full-range temporal graphs in \textsf{TC} (not necessarily cliques) are pivotable. %(I'm using definition of pivotable with respect to nodes)
\end{theorem}
\begin{proof}
  Let $e=\max\{e^-(v):v\in V\}$ and let $u\in V$ such that $e^-(u)=e$.
  By Lemma~\ref{lemma-fullrange}, every other node can reach $u$ by time $t$ and through $e^-(u)$ (since $u$ has no other incident edges before).
  As the graph is temporally connected, $u$ can also reach every other nodes, and since $e$ is its first edge, any temporal path from $u$ to another node either starts with $e$ or can be composed with it.
\end{proof}

\begin{theorem}
  \label{full-range-dismountable}
  Full-range temporal cliques are $\{1,2,3\}$-hop dismountable.
\end{theorem}
\begin{proof}
  If $|V|<3$, nothing needs to be proved, so let $|V|\ge 3$.
  Let $u\in V$ be such that $e^-(u)$ is maximum. By Lemma~\ref{lemma-fullrange}, every other node $v$ can reach $u$,
  starting from the edge $e^-(v)$ and arriving through $e^-(u)$ itself (since $u$ has no other edges before that and the lemma implies that $u$ is reached by time $\lambda(e^-(u))$).
  Thus, each node can delegate its emissions to $u$. Now, let $f=vw$ be the edge whose label is globally the maximum. Either $v$ or $w$ is not $u$. W.l.o.g., let $v$ be this node.
  Then $v$ can be reached by $w$ through $e^+(w)$. It is thus $k$-hop dismountable for some $k$, with respect to $u$ (delegation of emissions) and $w$ (delegation of receptions).
  From this, we can conclude by using Theorem~\ref{thm:threehopdismountable}.
\end{proof}

Unfortunately, the clique one obtains after $k$-hop dismounting a full-range clique may no longer be full-range, so Theorem~\ref{full-range-dismountable} does not imply that full-range cliques are \emph{recursively} $k$-hop dismountable. However, this is only a minor concern, as Theorem~\ref{full-range-pivotable} implies that full-range cliques admit $2n-3$ spanners. More generally, these results seem to indicate that having a large (compressed) lifetime facilitate the existence of sparser spanners. Further investigations are needed to determine to what extent this could be exploited in a more relaxed version, e.g. in terms of excluding instances that contain either large full-range sub-cliques or large subgraphs which are full-range and in \textsf{TC}.

\section{Conclusion}
\label{sec:conclusion}

In this paper, we revisited the concept of dismountability in a methodical manner, characterizing the properties that non $k$-hop dismountable cliques must have and showing that if a clique is $k$-hop dismountable for any $k$, then it is also $k$-hop dismountable with $k$ in $\{1,2,3\}$. We also showed that excluding $\{1,2\}$-hop dismountability is sufficient for reducing the spanner problem to its bi-clique version. Our results have several implications. One is that the existence of spanners of size $O(n \log n)$ can now be proven entirely using dismountability, establishing dismountability as a central technique for the problem. Another one is that any minimal counter-example to spanners of size $4n$ must possess all the structure of non $\{1,2,3\}$-hop dismountable cliques. We do not know yet if temporal cliques admit $O(n)$ spanners, thus the latter result might not appear very decisive. However, if the existence of $O(n)$ spanners were to be established, it would immediately play a significant role for pinning down the multiplicative constant, which is a relevant question in a structural context (perhaps less so when applied to time complexity). We also established a connection between dismountability and pivotability, both being well-known techniques that were so far unrelated. Finally, we investigated the properties of temporal graphs whose lifetime is maximum, namely full-range temporal graphs, showing that full-range cliques are dismountable and more generally, full-range temporally connected graphs are pivotable. This new class of temporal graphs may be of independent interest.

Obviously, the main open question remains whether temporal cliques admit linear-size spanners. We believe that our results are of interest beyond this question, due to the versatility of the dismountability principle. Another question is whether (and to which extent) dismountability could be relaxed. To motivate the question, observe that the type of spanner one gets by recursively $1$-hop dismounting a clique, when recursion works entirely, is a temporal graph whose footprint is a maximal $2$-degenerate graph (also known as $2$-arch graph). Intriguingly, preliminary experiments have shown that such spanners often exist even in cliques where the recursion fails (in fact, no counter-example was found). At the very least, this indicates that a more relaxed version of dismountability might be applicable more generally, which is worth investigating.

We conclude by giving an example of temporal clique that is neither pivotable nor $k$-hop dismountable (whatever $k$). This clique nonetheless admits a 2-arch spanner (thus, of size $2n - 3$) and another spanner of size $2n - 4$, the search of which we leave as a fun exercise to the reader.

\begin{figure}[h]
\centering
\begin{tikzpicture}[xscale=2.4,yscale=1.6]
 \tikzstyle{every node}=[inner sep=2pt,circle,black,fill=darkgray]
 \path (0,0) node (5){}; %u1 5
 \path (0,1) node (7){}; %u2 7
 \path (0,2) node (4){}; %u3 4
 \path (0,3) node (6){}; %u4 6

 \path (1,0) node (2){}; %v1 2
 \path (1,1) node (1){}; %v2 1
 \path (1,2) node (0){}; %v3 0
 \path (1,3) node (3){}; %v4 3

 \tikzstyle{every node}=[]
 \path (-1,2) node{\Large $V^+$};
 \path (2,2) node{\Large $V^-$};

 \tikzstyle{label}=[circle,fill=white,inner sep=1pt,font=\scriptsize,pos=.3]

 \draw[lightgray] (0) -- node[label,pos=.8]{10} (5);
 \draw[lightgray] (0) -- node[label,pos=.8]{11} (7);
 \draw[lightgray] (1) -- node[label]{16} (4);
 \draw[lightgray] (1) -- node[label]{17} (6);
 \draw[lightgray] (2) -- node[label]{12} (4);
 \draw[lightgray] (2) -- node[label]{13} (6);
 \draw[lightgray] (3) -- node[label]{14} (5);
 \draw[lightgray] (3) -- node[label]{15} (7);

 \draw[thick,red] (6) -- node[below=-2pt,pos=.1]{\tiny $+$} node[above=-2pt,pos=.92]{\Large $+$} node[label]{19} (0);
 \draw[thick,red] (4) -- node[below=-2pt,pos=.1]{\tiny $+$} node[below=-2pt,pos=.92]{\Large $+$} node[label]{20} (3);
 \draw[thick,red] (7) -- node[below=-2pt,pos=.1]{\tiny $+$} node[above=-2pt,pos=.92]{\Large $+$} node[label]{21} (2);
 \draw[thick,red] (5) -- node[below=-2pt,pos=.1]{\tiny $+$} node[below=-2pt,pos=.92]{\Large $+$} node[label]{18} (1);
 \draw[thick,green] (6) -- node[above=-4pt,pos=.12]{\Large $-$} node[above=-2pt,pos=.88]{\tiny $-$} node[label]{7} (3);
 \draw[thick,green] (4) -- node[below=-4pt,pos=.12]{\Large $-$} node[above=-2pt,pos=.88]{\tiny $-$} node[label]{8} (0);
 \draw[thick,green] (7) -- node[below=-5pt,pos=.12]{\Large $-$} node[above=-2pt,pos=.88]{\tiny $-$} node[label]{6} (1);
 \draw[thick,green] (5) -- node[below=-3pt,pos=.12]{\Large $-$} node[above=-2pt,pos=.88]{\tiny $-$} node[label]{9} (2);

 \tikzstyle{label}=[circle,fill=white,inner sep=1pt,font=\scriptsize,pos=.5] %changed default label pos for inner edges

 \tikzstyle{every edge}=[draw,bend left]
 \draw[thick,red,bend left=15] (5) edge node[left=-2pt,pos=.13]{\Large $+$} node[label]{26} (7);
 \draw[lightgray] (5) edge node[label]{22} (4);
 \draw[lightgray] (5) edge node[label]{25} (6);
 \draw[thick,red,bend left=15] (7) edge node[left=-2pt,pos=.87]{\Large $+$} node[label]{24} (4);
 \draw[thick,red] (7) edge node[left=-2pt,pos=.15]{\Large $+$} node[left=-2pt,pos=.85]{\Large $+$} node[label]{27} (6);
 \draw[lightgray,bend left=15] (4) edge node[label]{23} (6);

 \tikzstyle{every edge}=[draw,bend right]
 \draw[lightgray,bend left=15] (2) edge node[label]{3} (1);
 \draw[thick,green] (2) edge node[right=-2pt,pos=.18]{\Large $-$} node[label]{1} (0);
 \draw[lightgray] (2) edge node[label]{5} (3);
 \draw[thick,green,bend left=15] (1) edge node[right=-2pt,pos=.18]{\Large $-$} node[right=-2pt,pos=.88]{\Large $-$} node[label]{0} (0);
 \draw[lightgray] (1) edge node[label]{4} (3);
 \draw[thick,green,bend left=15] (0) edge node[right=-2pt,pos=.88]{\Large $-$} node[label]{2} (3);

\end{tikzpicture}
\caption{A temporal clique that is neither pivotable nor $k$-hop dismountable.}
\label{fig:ndis-npiv-8}
\end{figure}

\bibliographystyle{plain}
\bibliography{dismountability}
\end{document}